\documentclass[12pt,a4paper]{article} 
\usepackage{amsmath}
\usepackage{amssymb}
\usepackage{amsthm}
\usepackage{enumerate}
\usepackage{epsf}
\usepackage{psfrag}
\usepackage{mathrsfs}
\usepackage{hyperref}
\usepackage{geometry}
\DeclareGraphicsExtensions{.eps,.art,.ART,.ps}
\newcommand{\bbR}{\mathbb{R}}      
\newcommand{\bbN}{\mathbb{N}}      

\newcommand{\grad}{\operatorname{grad}}

\newcommand{\Real}{\operatorname{Re}}
\newcommand{\lb}{\label}
\newcommand{\be}{\begin{equation}}
\newcommand{\ee}{\end{equation}}
\newcommand{\ben}{\begin{eqnarray*}}
	\newcommand{\een}{\end{eqnarray*}}
\newcommand{\bea}{\begin{eqnarray}}
\newcommand{\eea}{\end{eqnarray}}

\newcounter{mnotecount}[section]

\renewcommand{\themnotecount}{\thesection.\arabic{mnotecount}}

\newcommand{\mnote}[1]
{\protect{\stepcounter{mnotecount}}$^{\mbox{\footnotesize $%
\!\!\!\!\!\!\,\bullet$\themnotecount}}$ \marginpar{
\raggedright\tiny\em $\!\!\!\!\!\!\,\bullet$\themnotecount: #1} }

\newtheorem{Thm}{Theorem}[section]

\begin{document}
\title{Linear relativistic thermoelastic rod}
\author{Anne T.\ Franzen and Jos\'e Nat\'ario\\
{\small Center for Mathematical Analysis, Geometry and Dynamical Systems,} \\
{\small Mathematics Department, Instituto Superior T\'ecnico,} \\
{\small Universidade de Lisboa, Portugal}}
\maketitle
\begin{abstract}
We derive and analyze the linearized hyperbolic equations describing a relativistic heat-conducting elastic rod. We construct a decreasing energy integral for these equations, compute the associated characteristic propagation speeds, and prove that the solutions decay in time by using a Fourier decomposition. For comparison purposes, we obtain analogous results for the classical system with heat waves, in which the finite propagation speed of heat is kept but the other relativistic terms are neglected, and also for the usual classical system.
\end{abstract}
\tableofcontents
%
%
%
%
%
\section{Introduction}\label{section0}
The motions of a relativistic rigid elastic rod, that is, a rod whose speed of sound is as large as possible (namely the speed of light), were studied in \cite{CN19}. Perhaps unsurprisingly, these motions were shown to be periodic, as might be expected from the fact that a purely elastic rod conserves energy. In the present paper we study a relativistic elastic rod settling to equilibrium by introducing energy dissipation and heat conduction. To avoid the ambiguity and complexity of the full nonlinear system, we derive a linear approximation of the equations of motion by linearizing the thermodynamic relations and the energy-momentum tensor around the equilibrium state, thus obtaining a linear hyperbolic system. Although it may be argued that the effects of the nonlinear terms that we discard are possibly more important than the effects of the relativistic terms that we keep, our main purpose here is precisely to understand the role played by these relativistic terms. For comparison purposes, we also discuss the classical system, as well as the classical system with heat waves, in which the finite propagation speed of heat is kept but the other relativistic terms are neglected.

Linear thermoelastic rods have been previously considered in the literature, both for the classical system and the classical system with heat waves. The equations for the classical system can be found in \cite{BW60, Carlson72}, and were thoroughly analyzed in \cite{Day85, Hansen92} (a general result for the relaxation of linear thermoelastic systems into equilibrium was also proved in \cite{Dafermos68}). The equations for the classical system with heat waves were introduced in \cite{LS67}, following the proposal by Cattaneo to use the telegrapher's equation as a model of heat conduction \cite{Cattaneo48}, and were studied in detail in \cite{Racke02}. Various versions of the relativistic thermoelasticity equations were proposed in \cite{Barrabes75, Kranys77, Palumbo85}, but the problem of the linear relativistic thermoelastic rod has not been addressed so far.

The organization of the paper is as follows. In Section~\ref{section1} we linearize the thermodynamic relations of the rod around its equilibrium state, expressing the coefficients in terms of the equilibrium temperature, number density, speed of sound and an adimensional parameter $\gamma$ which controls the coupling between heat and vibration. In Section~\ref{section2} we linearize the energy-momentum conservation laws and complement them with the linearization of the relativistic version of Cattaneo's heat conduction law, thus obtaining a linear hyperbolic system. In Section~\ref{section3} we derive an energy integral which is decreasing for solutions of the relativistic system. This integral is seen to reduce to the known expressions of the decreasing energy integrals for both the classical system \cite{Day85} and the classical system with heat waves \cite{Racke02}. In Section~\ref{section4} we compute the characteristic propagation speeds for both the relativistic system and the classical system with heat waves. We identify two different propagation speeds, which can be associated to mechanical waves and heat waves; interestingly, the speed of the mechanical waves is not exactly the (adiabatic) speed of sound, except at zero temperature. In Section~\ref{section5} we prove that the solutions decay in time by using a Fourier decomposition, and obtain the analogue results for the classical system with heat waves in Section~\ref{section6}, and for the classical system in Section~\ref{section7}. To allow for the comparison of these results we expand the eigenvalues about zero equilibrium temperature, instead of zero coupling parameter $\gamma$, as is usually done.
%
%
\section{Linear relativistic thermodynamics}\label{section1}
According to the principles of thermodynamics (see for instance \cite{Christodoulou95}), the rest energy density $\rho$ of the rod is a function of the number density $n$ and the entropy per particle $s$,
\begin{equation} \label{fundamental}
\rho = \rho(n,s).
\end{equation}
The pressure $p$ and the temperature $\theta$ are then obtained from
\begin{equation}
\lb{pressure1}
p = n \frac{\partial \rho}{\partial n} - \rho
\end{equation}
and
\begin{equation}
\lb{theta}
\theta = \frac1n \frac{\partial \rho}{\partial s},
\end{equation}
and the speed of sound $c$ is given by
\begin{equation}
\lb{c_square}
c^2 = \left( \frac{\partial p}{\partial \rho} \right)_s = n\frac{\partial^2 \rho}{\partial n^2} \left(\frac{\partial \rho}{\partial n}\right)^{-1}.
\end{equation}
To obtain linear equations we expand the fundamental thermodynamic equation to second order around an equilibrium state $(n_0,s_0)$. Setting
\begin{equation}
\hat{n} = n - n_0, \qquad \hat{s} = s - s_0,
\end{equation}
and choosing the mass/energy units so that the rest energy density of the rod in the equilibrium state is $\rho_0=1$, we have 
\begin{equation} \label{fundamental2}
\rho(\hat{n},\hat{s}) = 1 + \alpha \hat{n} + \beta \hat{s} + \frac12 \delta {\hat{n}}^2 + \varepsilon \hat{n} \hat{s} + \frac12 \varphi {\hat{s}}^2,
\end{equation}
where $\alpha$, $\beta$, $\delta$, $\varepsilon$ and $\varphi$ are constants to be specified shortly.
From \eqref{pressure1} together with  \eqref{fundamental2} we obtain, to linear order,
\bea
\lb{pressure2}
p(\hat{n}, \hat{s})&=&(n_0+\hat{n})(\alpha+\delta \hat{n}+\varepsilon \hat{s})-\rho\nonumber\\
&=&  n_0 \delta \hat{n}+(n_0 \varepsilon -\beta)\hat{s},
\eea
where we assumed that the pressure of the rod in the equilibrium state is zero:
\be
p_0 = n_0\alpha -1\,\,=\,\,0,
\lb{alpha}
\ee
implying
\be
\alpha=\frac1{n_0}.
\ee
Noting that, to zeroth order,
\bea
\frac{\partial \rho}{\partial \hat{n}}&=&\alpha,\\
\frac{\partial^2 \rho}{\partial \hat{n}^2}&=&\delta,
\eea
and recalling equation \eqref{c_square}, we obtain for the speed of sound in the equilibrium state
\be
c_0^2=\frac{n_0 \delta}{\alpha}\,\,=\,\,n_0^2 \delta,
\lb{delta}
\ee
whence
\be
\delta=\frac{c_0^2}{n_0^2}.
\ee
We further use \eqref{theta} to derive, to linear order,
\bea
\theta&=&\frac1{n_0+\hat{n}}(\beta+\varepsilon\hat{n}+\varphi\hat{s})\nonumber\\
&=&\left( \frac1{n_0}-\frac{\hat{n}}{n_0^2}\right) (\beta+\varepsilon\hat{n}+\varphi\hat{s})\nonumber\\
\lb{thetalinear}
&=&\frac{\beta}{n_0}+\left(\frac{\varepsilon}{n_0}-\frac{\beta}{n_0^2} \right) \hat{n}+\frac{\varphi}{n_0}\hat{s}.
\lb{beta}
\eea
Therefore,
\be
\beta=n_0\theta_0,
\ee
where $\theta_0$ is the temperature of the rod in the equilibrium state. Solving \eqref{thetalinear} for $\hat{s}$, we obtain
\bea
\lb{diff_entropy}
\hat{s}=\frac{n_0}{\varphi}\left[ \hat{\theta}+ \left(\frac{\theta_0-\varepsilon}{n_0}\right)\hat{n} \right] ,
\eea
where
\begin{equation}
\hat{\theta} = \theta - \theta_0.
\end{equation}
Starting from the linearized version of \eqref{fundamental2},
\bea
\lb{lin_fun}
\rho(\hat{n}, \hat{s})=1+\alpha\hat{n}+\beta\hat{s},
\eea
and using \eqref{alpha}, \eqref{beta} and \eqref{diff_entropy}, we can derive the energy density $\rho$ in terms of $\hat{n}$ and $\hat{\theta}$ as
\bea
\lb{rho_almost_final}
\rho(\hat{n}, \hat{\theta})&=&1 + \frac{\hat{n}}{n_0}+\frac{n_0^2 \theta_0}{\varphi}\left[  \hat{\theta}+\left(\frac{\theta_0-\varepsilon}{n_0}\right)\hat{n} \right] ,\nonumber\\
&=&1 + \hat{\theta}+ \left(\frac{1+\theta_0-\varepsilon}{n_0} \right) \hat{n},
\eea
where in the second step we have set (by choosing the temperature units conveniently)
\be
\frac{n_0^2 \theta_0}{\varphi}=1,
\ee
that is,
\be
\varphi=n_0^2\theta_0. \lb{epsilon}
\ee
Analogously, we can derive from the linearized equation \eqref{pressure2}, with the help of \eqref{delta}, \eqref{beta}, \eqref{diff_entropy} and \eqref{epsilon}, the formula
\bea
\lb{pressure_final}
p(\hat{n}, \hat{\theta})&=&\left(\frac{c_0^2-\gamma^2\theta_0}{n_0} \right) \hat{n}+\gamma\,\hat{\theta}
\eea
for the pressure, where 
\bea
\lb{gamma}
\gamma=\frac{\varepsilon-\theta_0}{\theta_0}.
\eea
Since we expect that increasing the temperature for a fixed particle number will increase the pressure, we impose the condition
\bea
\lb{con1}
\left( \frac{\partial p}{\partial \theta}\right)_n>0 \quad \Rightarrow\quad  \gamma >0.
\eea
Similarly, we also expect that increasing the particle number for a fixed temperature leads to higher pressure. Therefore, we assume that
\bea
\lb{con2}
\left( \frac{\partial p}{\partial n}\right)_{\theta}>0 \quad \Rightarrow\quad   c_0^2>\gamma^2\theta_0.
\eea
Using \eqref{gamma}, we can write equation \eqref{rho_almost_final} as
\bea
\lb{rho_final}
\rho(\hat{n}, \hat{\theta})=1 + \hat{\theta}+ \left(\frac{1-\gamma \theta_0}{n_0} \right) \hat{n},
\eea
Note that, as advertised, the constants $\alpha, \beta, \delta, \varepsilon, \varphi$ have been fixed as functions of $n_0, \theta_0, c_0^2$ and the adimensional parameter $\gamma$ by equations \eqref{alpha}, \eqref{delta}, \eqref{beta}, \eqref{epsilon} and \eqref{gamma}.
%
%
\section{Linearized equations of motion}\label{section2}
In the previous section we derived the final forms \eqref{pressure_final} and \eqref{rho_final} for the pressure and energy density as linearized functions of the number density and the temperature. In this section, we initiate the study of a solid rod moving along the $x$-axis of Minkowski's spacetime. To do that, we introduce a function $\lambda=\lambda(t,x)$ which gives the number of particles of the rod counted from a given endpoint. This function is constant along the worldlines of the rod's particles, and therefore serves as a measurement of its deformation; since its gradient is tangent to the rod's particles rest spaces, we have 
\bea
n=|\grad\lambda|=|d\lambda|.
\eea
Assuming that the initial endpoint sits at $x=0$ for the rod in equilibrium, we have the linearization
\be
\lb{lambda}
\lambda(x,t)=n_0x+ \hat{\lambda}(x, t),
\ee
and so
\be
d \lambda=n_0d x +d \hat{\lambda}=\left(n_0+\hat{\lambda}'\right)d x+\dot{ \hat{\lambda}}d t,
\ee
where the prime denotes differentiation with respect to $x$ and the dot differentiation with respect to $t$. Therefore, to linear order
\bea
|d \lambda|&=&\left(n_0^2+2n_0\hat{\lambda}'\right)^{\frac12}=n_0+\hat{\lambda}',
\eea
so that the deviation of the particle number from equilibrium is 
\bea
\lb{numberlambda}
\hat{n}&=&\hat{\lambda}'.
\eea
Since the four-velocity of the rod's particles must be orthogonal to $d \lambda$, the corresponding covector must be, to linear order,
\bea
\lb{nbrvelocity}
\tilde{ u}=\frac{1}{|d \lambda|}\left[-(n_0+\hat{\lambda}')d t-\dot{\hat{\lambda}} d x\right]=-d t-\frac{\dot{\hat{\lambda}}}{n_0}d x
\eea
(where the minus signs are needed  to make it future-pointing). The linearized four-velocity is then
\bea
\lb{velocity}
\vec{u}=\frac{\partial}{\partial t}-\frac{\dot{\hat{\lambda}}}{n_0}\frac{\partial}{\partial x},
\eea
and the linearized four-acceleration
\bea
\lb{acceleration}
\dot{\vec u}=\nabla_{\vec u} {\vec u}=-\frac{\ddot{\hat{\lambda}}}{n_0}\frac{\partial}{\partial x}.
\eea
Note that particle number conservation to linear order is automatic, since
\bea
\lb{part_n}
\mbox{div} (n\vec{u})&=&\mbox{div}\left[ (n_0+\hat{\lambda}')\left(\frac{\partial}{\partial t}-\frac{\dot{\hat{\lambda}}}{n_0}\frac{\partial}{\partial x} \right) \right] \nonumber \\
&=& \mbox{div}\left[ (n_0+\hat{\lambda}')\frac{\partial}{\partial t}-\dot{\hat{\lambda}}\frac{\partial}{\partial x} \right] \nonumber \\
&=& \dot{\hat{\lambda}}' - \dot{\hat{\lambda}}' = 0.
\eea
The relativistic version of Cattaneo's heat conduction law, obtained by Israel and Stewart in \cite{IS79} (see also \cite{Hayward99}), gives the heat flow vector field as
\bea
\lb{heat_flow}
\vec{q}=-\kappa\left(\nabla \theta+\theta { \dot{\vec u}}+bn\theta^2{ \dot{\vec q}}\right)^{\perp}
\eea
(projected orthogonally to the velocity), where $\kappa$ is the thermal conductivity, $\nabla \theta$ is the temperature gradient and $b$ is a constant. Linearizing with respect to the equilibrium state reduces this quantity to a scalar,
\bea
\vec{q}=q \frac{\partial}{\partial x},
\eea
with
\bea
q=-\kappa\left(\hat{\theta}'-\frac{\theta_0}{n_0} \ddot{\hat{\lambda}}+b n_0\theta_0^2 \dot{q}\right).
\eea
We rewrite this equation as 
\bea
\lb{heat_speed_eq}
\frac{\kappa}{c_h^2} \dot{q}+q+\kappa\left(\hat{\theta}'-\frac{\theta_0}{n_0} \ddot{\hat{\lambda}}\right)=0, 
\eea
where, as we shall see, the new constant
\bea
\lb{speed_heat}
c_h= \frac{1}{\theta_0 \sqrt{b n_0}}
\eea
is related to the speed of heat waves.

To linear order, the energy-momentum tensor is then given by
\bea
\lb{energymt}
T^{\mu \nu}&=&(\rho +p)u^{\mu}u^{\nu}+pg^{\mu \nu}+u^{\mu}q^{\nu}+u^{\nu}q^{\mu},\\
&=&\begin{bmatrix}
\rho      & -(\rho+p)\dot{\hat{\lambda}}/n_0+q  \\
-(\rho+p)\dot{\hat{\lambda}}/n_0+q      & p \\
\end{bmatrix},\\
&=&\begin{bmatrix}
1+\hat{\theta}+(1-\gamma\theta_0)\hat{\lambda}' /n_0    & -\dot{\hat{\lambda}}/n_0+q  \\
	-\dot{\hat{\lambda}}/n_0+q      & (c_0^2-\gamma^2\theta_0)\hat{\lambda}'/n_0+\gamma\hat{\theta} 
\end{bmatrix},
\eea
where we have used \eqref{pressure_final}, \eqref{rho_final}, \eqref{numberlambda} and \eqref{velocity}.
We can now obtain the equations of motion of the rod from energy-momentum conservation,
\bea
\nabla_{\mu} T^{\mu \nu}=0,
\eea
leading to the following system of equations:
\bea
\lb{disp_system1}
&&\dot{\hat{\theta}}+(1-\gamma\theta_0)\dot{\hat{\lambda}}'/n_0-\dot{\hat{\lambda}}'/n_0+q' = 0,\\
\lb{disp_system2}
&&-\ddot{\hat{\lambda}}/n_0+\dot{q}+(c_0^2-\gamma^2\theta_0)\hat{\lambda}''/n_0+\gamma\hat{\theta}'=0.
\eea
These equations are to be taken together with \eqref{heat_speed_eq}. 

The equilibrium position of the particle at event $(t,x)$ (i.e.\ the position coordinate in the Eulerian description of the motion) is simply
\bea
\xi(t,x) = \frac{\lambda(t,x)}{n_0} = x + \frac{\hat{\lambda}(t,x)}{n_0}.
\eea
Therefore, the quantity
\be \label{wdisplacement}
w(x,t)=-\frac{\hat{\lambda}(x,t)}{n_0}
\ee
represents, to linear order, the displacement of a given particle with respect to its equilibrium position. Using this variable, equations \eqref{heat_speed_eq}, \eqref{disp_system1} and \eqref{disp_system2} then read
\bea
\lb{displacement1}
&&\dot{\hat{\theta}}+\gamma \theta_0\dot{w}'+q'=0,\\
\lb{displacement2}
&&\ddot{w}-c_s^2w''+\gamma\hat{\theta}'+\dot{q}=0,\\
\lb{displacement3}
&&\frac{\kappa}{c_h^2}\dot{q}+q+\kappa(\hat{\theta}'+\theta_0\ddot{w})=0,
\eea
where we have set
\bea \label{cssquared}
c_s^2 = c_0^2-\gamma^2\theta_0.
\eea
Notice that \eqref{con2} guarantees that $c_s^2$ is indeed positive; we will see later that this constant is related to the speed of mechanical waves. Using \eqref{wdisplacement} and \eqref{cssquared} we can rewrite equation~\eqref{pressure_final} as
\bea
\lb{pressure}
p=\gamma \hat{\theta}-c_s^2w'.
\eea
Since the Eulerian position coordinate $\xi$ differs from the Lagrangian position coordinate $x$ by a term of first order, and since all variables in equations~\eqref{displacement1}-\eqref{displacement3} above are also of first order, we can take $x$ to be either the Lagrangian or the Eulerian position coordinate in these equations. Since the boundary conditions are much simpler for the latter, we make that choice from this point on.

In comparison with the classical equations for a thermoelastic rod,
\bea
\lb{classical1}
&&\dot{\hat{\theta}}+\gamma \theta_0\dot{w}'+q'=0,\\
\lb{classical2}
&&\ddot{w}-c_s^2w''+\gamma\hat{\theta}'=0,\\
\lb{classical3}
&&q+\kappa\hat{\theta}'=0
\eea 
(see for instance \cite[\S 1]{Hansen92}), equations~\eqref{displacement1}-\eqref{displacement3} contain three additional terms: the term $\dot{q}$ in equation~\eqref{displacement2}, which represents the force due to variations in the momentum of heat; the term $\kappa\theta_0\ddot{w}$ in equation \eqref{displacement3}, which can be understood as a consequence of the redshift effect (see for instance \cite{SV19}); and the Cattaneo term $\kappa\dot{q}/c_h^2$ in equation \eqref{displacement3}, which allows for the fact that the heat propagates with finite speed. It is the role played by these relativistic terms that we wish to understand. Neglecting the first two terms while keeping the third corresponds to adding {\it heat waves} to the classical model:
\bea
\lb{heatwaves1}
&&\dot{\hat{\theta}}+\gamma \theta_0\dot{w}'+q'=0,\\
\lb{heatwaves2}
&&\ddot{w}-c_s^2w''+\gamma\hat{\theta}'=0,\\
\lb{heatwaves3}
&&\frac{\kappa}{c_h^2}\dot{q}+q+\kappa\hat{\theta}'=0.
\eea
This model has been considered before in \cite[\S 3]{LS67}. As noted in \cite{JP89}, equation \eqref{heatwaves3} is equivalent (assuming enough regularity and boundedness of $q$ and $\hat\theta'$) to 
\bea
q(t,x)=- c_h^2 \int_{-\infty}^t \hat{\theta}'(s,x) e^{\frac{c_h^2}{\kappa}(s-t)} ds,
\eea
that is, it corresponds to a version of the Fourier law \eqref{classical3} averaged over the past with an exponentially decreasing weight. Note also that
\bea
\lim_{c_h \to + \infty} c_h^2 \int_{-\infty}^t \hat{\theta}'(s,x) e^{\frac{c_h^2}{\kappa}(s-t)} ds = \kappa \hat{\theta}'(t,x),
\eea
as might be expected: The classical model is obtained by assuming that the speed of heat $c_h$ is infinite.
%
%
\section{Decreasing energy integral}\label{section3}
We choose as boundary conditions either
\bea
\lb{bd_con}
w(0)=w(L)=q(0)=q(L)=0,
\eea
or
\bea
\lb{bd_con2}
p(0)=p(L)=q(0)=q(L)=0,
\eea
where $L$ is the length of the rod and the pressure $p$ is given by equation \eqref{pressure}. These correspond to assuming that the rod is insulated, so that no heat flows into it, and that the endpoints are either clamped or completely free to move. 

\begin{Thm}
If we take as boundary conditions either \eqref{bd_con} or \eqref{bd_con2} then the relativistic system \eqref{displacement1}--\eqref{displacement3} admits the decreasing energy integral\footnote{This should not be confused with the actual energy integral, $$E = \int_{0}^L T^{00}(1+w')dx = \int_{0}^L (1 + \hat{\theta}+\gamma \theta_0 w') dx,$$ where the calculation is to linear order and the factor $(1+w')$ is necessary because we are regarding $x$ as an Eulerian coordinate; this integral is conserved by \eqref{displacement1} and either \eqref{bd_con} or \eqref{bd_con2}, and implies the conservation of $E_1 = \int_{0}^L (\hat{\theta}+\gamma \theta_0 w') dx$ ($= \int_{0}^L \hat{\theta} dx$ if the boundary condition \eqref{bd_con} is used). It is easily seen from \eqref{diff_entropy}, \eqref{gamma} and \eqref{wdisplacement} that $E_1$ is also proportional to the linear perturbation in the total entropy, which we expect to be constant because there is no heat flow into the rod.} 
\bea
\lb{decr_energy}
\tilde{E}=\int\limits^L_0\left( \frac{\theta_0}{2} \dot{w}^2+\frac{\theta_0 c_s^2}{2} {w'}^2+\frac12 {\hat{\theta}}^2+\frac{1}{2c_h^2}q^2+\theta_0\dot{w}q\right)d x.  
\eea

\end{Thm}

\begin{proof}
Taking the derivative with respect to $t$ yields
\bea
\lb{dt_energy}
\frac{d \tilde{E}}{d t}&=&\int\limits^L_0\left( \theta_0\dot{w}\ddot{w} + \theta_0 c_s^2{w'}\dot{w'} + {\hat{\theta}}\dot{\hat{\theta}}+\frac{1}{c_h^2}q\dot{q}+\theta_0\ddot{w}q+\theta_0\dot{w}\dot{q}\right)d x\nonumber \\
&=&   \int\limits^L_0\left(\left[ \theta_0 \dot{w}(c_s^2w' -\gamma \hat{\theta}) \right] '-(\hat{\theta} q)'-\frac{1}{\kappa} q^2\right)d x\nonumber \\
&=& -\frac{1}{\kappa}  \int\limits^L_0 q^2d x \leq0.
\eea
\end{proof}

Similar calculations show that the classical system with heat waves \eqref{heatwaves1}--\eqref{heatwaves3} admits the decreasing energy integral
\bea
\tilde{E}=\int\limits^L_0\left( \frac{\theta_0}{2} \dot{w}^2+\frac{\theta_0 c_s^2}{2} {w'}^2+\frac12 {\hat{\theta}}^2+\frac{1}{2c_h^2}q^2\right)d x
\eea
(already found in \cite{Racke02}), whereas the classical system \eqref{classical1}--\eqref{classical3} admits the decreasing energy integral
\bea
\tilde{E}=\int\limits^L_0\left( \frac{\theta_0}{2} \dot{w}^2+\frac{\theta_0 c_s^2}{2} {w'}^2+\frac12 {\hat{\theta}}^2\right)d x
\eea
(already found in \cite{Day85}).
%
%
\section{Characteristic propagation speeds}\label{section4}
By introducing the new variables $s=w'$ and $v=\dot{w}$, the relativistic system \eqref{displacement1}--\eqref{displacement3} can be cast in as a first order system:
\bea
\lb{cp_system1}
&& \dot{s}-v'=0,\\
&& \dot{v}-c_s^2s'+ \gamma \hat{\theta}'+\dot{q}=0,\\
&&\dot{\hat{\theta}}+q'+\gamma \theta_0 v'=0,\\
\lb{cp_system4}
&& \frac{\kappa}{c_h^2}\dot{q}+q+\kappa \hat{\theta}'+\kappa \theta_0 \dot{v}=0.
\eea
The system can be written in matrix form as
\bea
\begin{bmatrix}
1      & 0&0&0  \\ 
0   & 1&0&1 \\
0&0&1&0\\
0& \kappa \theta_0& 0& \frac{\kappa}{c_h^2}
\end{bmatrix}
\begin{bmatrix}
\dot{s}  \\ 
\dot{v} \\
\dot{\hat{\theta}}\\
\dot{q}
\end{bmatrix}
+
\begin{bmatrix}
0      & -1 & 0 & 0  \\ 
-c_s^2  & 0 & \gamma & 0 \\
0 & \gamma \theta_0 & 0 & 1\\
0 & 0 & \kappa& 0
\end{bmatrix}
\begin{bmatrix}
{s}'  \\ 
{v}' \\
{\hat{\theta}}'\\
{q}'
\end{bmatrix}
=  \begin{bmatrix}
0  \\ 
0 \\
0\\
-q
\end{bmatrix}.
\eea
As is well known \cite{Alinhac099}, the characteristic propagation speeds are given by the eigenvalues of the matrix
\bea
\mathbb{A}=
\begin{bmatrix}
1      & 0 & 0 &0  \\ 
0   & 1 & 0 & 1 \\
0 & 0 & 1 & 0\\
0& \kappa \theta_0& 0& \frac{\kappa}{c_h^2}
\end{bmatrix}^{-1}
\begin{bmatrix}
0      & -1 & 0 & 0  \\ 
-c_s^2  & 0 & \gamma & 0 \\
0 & \gamma \theta_0 & 0 & 1\\
0 & 0 & \kappa& 0
\end{bmatrix},
\eea
which can be found by solving
\begin{align}
\lb{un_rel_eigen}
& \det\left(\mathbb{A}-\lambda\cdot\mathbb{I}\right) = 0 \Leftrightarrow \nonumber \\
& (1-c_h^2\theta_0)\lambda^4-\left(c_h^2+c_s^2+\gamma^2\theta_0-2\gamma c_h^2\theta_0 \right)\lambda^2+c_s^2c_h^2 = 0.
\end{align}
It is interesting to note that for $\theta_0=0$ the eigenvalues are given by
\bea
\lambda^2 = c_s^2 \qquad \text{ or } \qquad \lambda^2 = c_h^2,
\eea
corresponding, as one might expect, to the speed of sound and speed of heat along the rod. To obtain the first order correction\footnote{From equation \eqref{rho_final} it is clear that in our units temperature is given by the ratio between the thermal energy density and the rest density of the rod, so that one expects $\theta_0 \ll 1$. We expand the eigenvalues about zero equilibrium temperature $\theta_0$, instead of zero coupling parameter $\gamma$ (as is usually done), to facilitate the comparison with the classical system with heat waves, where $\gamma \ll 1$ would compete with $c_h^2, c_s^2 \ll 1$.} in $\theta_0$ we differentiate equation~\eqref{un_rel_eigen} at $\theta_0=0$,
\bea
-c_h^2 \lambda^4 + 2\lambda^2 \frac{\partial \lambda^2}{\partial \theta_0} - \left(\gamma^2-2\gamma c_h^2\right) \lambda^2 - \left(c_h^2+c_s^2\right)\frac{\partial \lambda^2}{\partial \theta_0} = 0,
\eea
obtaining 
\bea
&&\lambda^2=c_s^2 \qquad \Rightarrow \qquad \frac{\partial \lambda^2}{\partial \theta_0}=\frac{c_s^2(\gamma^2 - 2\gamma c_ h^2 + c_h^2c_s^2)}{c_s^2-c_h^2}; \\
&&\lambda^2=c_h^2 \qquad \Rightarrow \qquad \frac{\partial \lambda^2}{\partial \theta_0}=\frac{c_h^2(\gamma^2 - 2\gamma c_ h^2 + c_h^4)}{c_h^2-c_s^2}.
\eea
Therefore, we get for the eigenvalues
\bea
\lambda_s^2 &=& c_s^2\left[1+\frac{(\gamma^2 - 2\gamma c_h^2 + c_h^2c_s^2)\theta_0}{c_s^2-c_h^2}+\ldots\right]; \label{c_s_squared_relativistic} \\
\lambda_h^2 &=& c_h^2\left[1+\frac{(\gamma^2 - 2\gamma c_h^2 + c_h^4)\theta_0}{c_h^2-c_s^2}+\ldots\right]. \label{c_h_squared_relativistic}
\eea
We have then proved the following result.

\begin{Thm}
The relativistic system \eqref{displacement1}--\eqref{displacement3} is a linear hyperbolic system with characteristic propagation speeds given as a power series in $\theta_0$ by
\bea
\lambda_s &=& \pm c_s\left[1+\frac{(\gamma^2 - 2\gamma c_h^2 + c_h^2c_s^2)\theta_0}{2(c_s^2-c_h^2)}+\ldots\right]; \label{c_s_relativistic} \\
\lambda_h &=& \pm c_h\left[1+\frac{(\gamma^2 - 2\gamma c_h^2 + c_h^4)\theta_0}{2(c_h^2-c_s^2)}+\ldots\right]. \label{c_h_relativistic}
\eea
\end{Thm}

Assuming that the zero-temperature speed of heat is bigger than the zero-temperature speed of sound, $c_h>c_s$, the speed of heat increases with increasing temperature, whereas the speed of sound increases for 
\be
c_h^2 - c_h \sqrt{c_h^2-c_s^2} < \gamma < c_h^2 + c_h \sqrt{c_h^2-c_s^2}
\ee
and decreases otherwise. Conversely, if the zero-temperature speed of heat is smaller than the zero-temperature speed of sound, $c_h<c_s$, then the speed of heat decreases and the speed of sound increases with increasing temperature.

Similarly, for the classical system with heat waves we obtain
\bea
\lb{cp_class_system1}
&& \dot{s}-v'=0,\\
&& \dot{v}-c_s^2s'+ \gamma \hat{\theta}'=0,\\
&&\dot{\hat{\theta}}+q'+\gamma \theta_0 v'=0,\\
\lb{cp_class_system4}
&& \frac{\kappa}{c_h^2}\dot{q}+q+\kappa \hat{\theta}'=0,
\eea
so that we can write 
\bea
\lb{sys_hw}
\begin{bmatrix}
\dot{s} \\
\dot{v} \\
\dot{\hat{\theta}}\\
\dot{q}
\end{bmatrix}=\begin{bmatrix}
0      & 1&0&0 \\
c_s^2  & 0&-\gamma &0 \\
0&-\gamma \theta_0&0&-1\\
0& 0 &- {c_h^2}& 0
\end{bmatrix}
\begin{bmatrix}
{s'} \\
{v'} \\
{\hat{\theta}'}\\
{q'}
\end{bmatrix}
+
\begin{bmatrix}
0 \\
0 \\
0\\
-\frac{c_h^2}{\kappa}q
\end{bmatrix}.
\eea
The characteristic propagation speeds are given by the eigenvalues of the matrix
\bea
\lb{A2}
\mathbb{A}=
\begin{bmatrix}
	0      & 1&0&0 \\
	c_s^2  & 0&-\gamma &0 \\
	0&-\gamma \theta_0&0&-1\\
	0& 0 &- {c_h^2}& 0
\end{bmatrix},
\eea
and can be found by solving
\bea
\lb{un_class_eigen}
\det\left( 
\mathbb{A}
-\lambda\cdot  \mathbb{I}\right)
=(\lambda^2-c_s^2)(\lambda^2-c_h^2)-\gamma^2 \theta_0 \lambda^2=0
\eea
(which can be seen as the limit of \eqref{un_rel_eigen} when $c_h^2 \theta_0 \ll 1$ and $c_h^2 \ll \gamma$). Again, for $\theta_0=0$ the eigenvalues are given by
\bea
\lambda^2 = c_s^2 \qquad \text{ or } \qquad \lambda^2 = c_h^2,
\eea
corresponding to the speed of sound and speed of heat along the rod. To obtain the first order correction in $\theta_0$ we differentiate equation~\eqref{un_class_eigen} at $\theta_0=0$,
\bea
(\lambda^2-c_s^2)\frac{\partial \lambda^2}{\partial \theta_0} + (\lambda^2-c_h^2)\frac{\partial \lambda^2}{\partial \theta_0} - \gamma^2 \lambda^2 = 0,
\eea
obtaining 
\bea
&&\lambda^2=c_s^2 \qquad \Rightarrow \qquad \frac{\partial \lambda^2}{\partial \theta_0}=\frac{c_s^2 \gamma^2}{c_s^2-c_h^2}; \\
&&\lambda^2=c_h^2 \qquad \Rightarrow \qquad \frac{\partial \lambda^2}{\partial \theta_0}=\frac{c_h^2 \gamma^2}{c_h^2-c_s^2}.
\eea
Therefore, we get for the eigenvalues
\bea
\lambda_s^2 &=& c_s^2\left(1+\frac{\gamma^2\theta_0}{c_s^2-c_h^2}+\ldots\right); \label{c_s_squared_classical} \\
\lambda_h^2 &=& c_h^2\left(1+\frac{\gamma^2\theta_0}{c_h^2-c_s^2}+\ldots\right). \label{c_h_squared_classical}
\eea
We have then proved the following result.

\begin{Thm}
The classical system with heat waves \eqref{heatwaves1}--\eqref{heatwaves3} is a linear hyperbolic system with characteristic propagation speeds given as a power series in $\theta_0$ by
\bea
\lambda_s &=& \pm c_s\left[1+\frac{\gamma^2\theta_0}{2(c_s^2-c_h^2)}+\ldots\right]; \label{c_s_classical} \\
\lambda_h &=& \pm c_h\left[1+\frac{\gamma^2\theta_0}{2(c_h^2-c_s^2)}+\ldots\right].
\eea
\end{Thm}

This result can be obtained from equations \eqref{c_s_relativistic} and \eqref{c_h_relativistic}, as one might expect, in the limit $c_s^2, c_h^2 \ll \gamma$. Assuming that the zero-temperature speed of heat is bigger than the zero-temperature  speed of sound, $c_h>c_s$, the speed of sound decreases and the speed of heat increases with increasing temperature. Conversely, if the zero-temperature speed of heat is smaller than the zero-temperature speed of sound, $c_h<c_s$, then the speed of sound increases and the speed of heat decreases with increasing temperature.

It is interesting to note that the speed of the mechanical waves is not exactly the nominal speed of sound \eqref{c_square}, both for the relativistic system and the classical system with heat waves, except at zero temperature. Indeed, if we substitute $c_s^2$ as given in \eqref{cssquared} into \eqref{c_s_squared_relativistic} or \eqref{c_s_squared_classical} we obtain\footnote{Note that the expansions \eqref{c_s_squared_relativistic}, \eqref{c_h_squared_relativistic}, \eqref{c_s_squared_classical} and \eqref{c_h_squared_classical} were perfomed assuming fixed $c_s$ (hence variable $c_0$).}
\bea
\lambda_s^2 &=& c_0^2\left[1+\frac{\left[\gamma^2(1+c_h^2-c_0^2) - 2\gamma c_h^2 + c_h^2c_0^2\right]\theta_0}{c_0^2-c_h^2}+\ldots\right]
\eea
or
\bea
\lambda_s^2 &=& c_0^2\left[1+\frac{\gamma^2(1+c_h^2-c_0^2)\theta_0}{c_0^2-c_h^2}+\ldots\right].
\eea
The reason for this discrepancy is that the speed of sound applies to adiabatic processes, whereas our systems feature energy dissipation and heat flow.

%
%
\section{Solving the relativistic system}\label{section5}
We will now use a Fourier decomposition to prove that the solutions of the relativistic system \eqref{displacement1}-\eqref{displacement3} decay in time. This is most easily done for the boundary conditions \eqref{bd_con}: If, for simplicity, we use our remaining freedom in the choice of units to set the length of the rod to $L=\pi$, then these boundary conditions are automatically enforced by choosing
\bea
\lb{thetacos}
\hat{\theta}(x, t)&=&\sum^{\infty}_{n=1} a_n(t) \cos(n x),\\
\lb{wsin}
w(x, t)&=&\sum^{\infty}_{n=1} b_n(t) \sin(n x),\\
\lb{qsin}
q(x, t)&=&\sum^{\infty}_{n=1} d_n(t) \sin(n x).
\eea
As is well known, the functions $\hat{\theta}$, $w$ and $q$ admit uniformly and absolutely convergent Fourier expansions of this kind if we assume the natural regularities, namely $\hat{\theta}$ and $q$ of class $C^1$ and $w$ of class $C^2$; we can assume that the coefficient $a_0$ in the expansion of $\hat{\theta}$ vanishes (or, equivalently, that $\hat{\theta}$ has zero spatial average) by subtracting off the trivial solution $\hat{\theta} \equiv a_0$, $w \equiv q \equiv 0$. Substituting this ansatz into \eqref{displacement1}-\eqref{displacement3}, we obtain, for each $n \in \bbN$, the system of ordinary differential equations
\bea
\dot{a}_n(t)&=&-n\gamma\theta_0\dot{b}_n(t) - nd_n(t),\\
\ddot{b}_n(t)&=&-n^2c_s^2b_n(t)+n\gamma a_n(t)-\dot{d}_n(t),\\
\dot{d}_n(t)&=&-\frac{c_h^2}{\kappa}d_n(t)+nc_h^2a_n(t)-c_h^2\theta_0\ddot{b}_n(t),
\eea
whence
\bea
\ddot{b}_n(t)&=&\frac{n(\gamma-c_h^2)}{1-c_h^2\theta_0}a_n(t)-\frac{n^2c_s^2}{1-c_h^2\theta_0}b_n(t)+\frac{c_h^2}{\kappa(1-c_h^2 \theta_0)}d_n(t),\\
\dot{d}_n(t)&=& \frac{nc_h^2(1-\gamma \theta_0)}{1-c_h^2\theta_0}a_n(t)+\frac{n^2c_h^2 c_s^2 \theta_0}{1-c_h^2\theta_0}b_n(t)-\frac{c_h^2}{\kappa(1-c_h^2 \theta_0)}d_n(t).
\eea
We transform this into a first order system by setting
\bea
\dot{b}_n(t)&=&f_n(t),
\eea
whence
\bea
\dot{a}_n(t)&=&-nd_n(t)-n\gamma\theta_0 f_n(t),\\
\dot{d}_n(t)&=&\frac{nc_h^2(1-\gamma \theta_0)}{1-c_h^2\theta_0}a_n(t)+\frac{n^2 c_h^2 c_s^2 \theta_0}{1-c_h^2\theta_0}b_n(t)-\frac{c_h^2}{\kappa(1-c_h^2 \theta_0)}d_n(t),\\
\dot{f}_n(t)&=&\frac{n(\gamma-c_h^2)}{1-c_h^2\theta_0}a_n(t)-\frac{n^2c_s^2}{1-c_h^2\theta_0}b_n(t)+\frac{c_h^2}{\kappa(1-c_h^2 \theta_0)}d_n(t).
\eea
In matrix form, the system reads
\bea
\begin{bmatrix}
\dot{a}_n  \\ 
\dot{b}_n  \\
\dot{d}_n  \\
\dot{f}_n 
\end{bmatrix}
=
\begin{bmatrix}
0    &0&- n &-n\gamma\theta_0 \\ 
0  & 0&0 &1 \\
\frac{nc_h^2(1-\gamma \theta_0)}{1-c_h^2\theta_0} & \frac{n^2 c_h^2 c_s^2 \theta_0}{1-c_h^2\theta_0}&-\frac{c_h^2}{\kappa(1-c_h^2 \theta_0)}&0 \\
\frac{n(\gamma-c_h^2)}{1-c_h^2 \theta_0}&-\frac{n^2c_s^2}{1-c_h^2 \theta_0}&\frac{c_h^2}{\kappa(1-c_h^2 \theta_0)}&0
\end{bmatrix}
\begin{bmatrix}
{a}_n  \\ 
{b}_n  \\
{d}_n  \\
{f}_n 
\end{bmatrix}.
\eea
The eigenvalues of the system's matrix are obtained by solving the polynomial equation
\begin{align} \label{eigen_rel_sys}
& \kappa  (1-c_h^2\theta_0)\lambda^4 + c_h^2 \lambda^3 + n^2 \kappa (c_h^2 + c_s^2 + \gamma^2 \theta_0 - 2 c_h^2 \gamma \theta_0) \lambda^2 \nonumber \\
& + n^2 c_h^2 (c_s^2 + \gamma^2 \theta_0) \lambda  + n^4 \kappa c_h^2 c_s^2 = 0.
\end{align}
It is interesting to note that for $\theta_0=0$ the eigenvalues are given by
\bea
\lambda= \frac{-c_h^2\pm c_h \sqrt{c_h^2- 4 n^2\kappa^2}}{2\kappa} \qquad \text{ or } \qquad  \lambda=\pm in c_s.
\eea
These correspond to decaying modes, associated to the diffusion constant $\kappa$ and the speed of heat $c_h$, and non-decaying oscillating modes, associated to the speed of sound $c_s$.

To obtain the first order correction\footnote{We expand the eigenvalues about zero equilibrium temperature $\theta_0$, instead of zero coupling parameter $\gamma$ (as is usually done), to facilitate the comparison with the classical system with heat waves, where $\gamma \ll 1$ would compete with $c_h^2, c_s^2 \ll 1$.} in $\theta_0$ we differentiate equation~\eqref{eigen_rel_sys} at $\theta_0=0$,
\begin{align}
&\left[ 4 \kappa\lambda^3 + 3c_h^2 \lambda^2 + 2n^2 \kappa (c_h^2 + c_s^2 ) \lambda + n^2 c_h^2 c_s^2  
\right] \frac{\partial \lambda}{\partial \theta_0} \nonumber\\
&-\kappa c_h^2 \lambda^4 + n^2 \kappa (\gamma^2 - 2c_h^2\gamma) \lambda^2 + n^2c_h^2\gamma^2\lambda = 0.
\end{align}
For $\lambda=\pm in c_s$ we obtain
\begin{align}
& \left[- 2 n^2 c_h^2 c_s^2 \pm 2i n^3\kappa c_s (c_h^2 - c_s^2) \right] \frac{\partial \lambda}{\partial \theta_0} \nonumber \\
& - n^4\kappa c_s^2 (\gamma^2 - 2 c_h^2 \gamma + c_h^2 c_s^2) \pm i n^3 c_h^2 c_s \gamma^2 = 0,
\end{align}
and so
\begin{align} \lb{partial_lambda_rel}
\frac{\partial \lambda}{\partial \theta_0} & = \frac{n^2\kappa c_s (\gamma^2 - 2 c_h^2 \gamma + c_h^2 c_s^2) \mp i n c_h^2 \gamma^2}{- 2 c_h^2 c_s \pm 2i n\kappa (c_h^2 - c_s^2)} \nonumber \\
& = \frac{n^2 \kappa c_h^4 (- \gamma^2 + 2 c_s^2 \gamma - c_s^4) \pm in c_h^4 c_s \gamma^2}{2[c_h^4 c_s^2 + n^2\kappa^2 (c_h^2 - c_s^2)^2]} \nonumber \\
& \mp \frac{i n^3 \kappa^2 c_s (c_h^2 - c_s^2) (\gamma^2 - 2 c_h^2 \gamma + c_h^2 c_s^2)}{2[c_h^4 c_s^2 + n^2\kappa^2 (c_h^2 - c_s^2)^2]}.
\end{align}
Note that as $n \to \infty$ \eqref{partial_lambda_rel} becomes
\bea
\frac{\partial \lambda}{\partial \theta_0} \sim \mp \frac{i n c_s (\gamma^2 - 2 c_h^2 \gamma + c_h^2 c_s^2)}{2(c_h^2 - c_s^2)},
\eea
and so
\bea
\lambda \sim \pm i n c_s \left[ 1 - \frac{ (\gamma^2 - 2 c_h^2 \gamma + c_h^2 c_s^2) \theta_0}{2(c_h^2 - c_s^2)} + \ldots \right],
\eea
in agreement with \eqref{c_s_relativistic}.

\begin{Thm} \label{Thmrel}
All modes in the Fourier expansion \eqref{thetacos}-\eqref{qsin} for the solutions of the relativistic system \eqref{displacement1}--\eqref{displacement3} with boundary conditions \eqref{bd_con} decay exponentially in time for $\theta_0 > 0$, except at the critical temperature $\theta_0 = (\gamma - c_s^2) / \gamma^2$ if $\gamma > c_s^2$. Consequently, given initial conditions $\hat{\theta}(0,x)$, $w(0,x)$, $\dot{w}(0,x)$ and $q(0,x)$ of class $C^1$, with $\int_0^\pi \hat{\theta}(0,x) dx = 0$, we have
\be
\lim_{t \to +\infty} \hat{\theta}(x, t) = \lim_{t \to +\infty} w(t,x) = \lim_{t \to +\infty} q(t,x) = 0
\ee
for $\theta_0 > 0$, except at the critical temperature $\theta_0 = (\gamma - c_s^2) / \gamma^2$ if $\gamma > c_s^2$.
\end{Thm}

\begin{proof}
Since
\bea
\Real \left( \frac{\partial \lambda}{\partial \theta_0} \right) = - \frac{n^2 \kappa c_h^4 (\gamma - c_s^2)^2}{2[c_h^4 c_s^2 + n^2\kappa^2 (c_h^2 - c_s^2)^2]} < 0,
\eea
the eigenvalues (which depend continuously on $\theta_0$) enter the half-plane $\Real(\lambda)<0$ as $\theta_0$ increases. If we set $\lambda=iy$ in equation~\eqref{eigen_rel_sys}, with $y \in \bbR$, then we obtain from the imaginary part
\be
y = 0 \qquad \text{ or } \qquad y^2 = n^2(c_s^2+\gamma^2 \theta_0).
\ee
Now equation~\eqref{eigen_rel_sys} is never satisfied in the first case, and is satisfied in the second case if and only if
\be
\theta_0 = 0 \qquad \text{ or } \qquad \theta_0 = \frac{\gamma - c_s^2}{\gamma^2}.
\ee
If $\gamma > c_s^2$ then the eigenvalues might leave the half-plane $\Real(\lambda)<0$ as the critical temperature $\theta_0 = (\gamma - c_s^2) / \gamma^2$ was reached;\footnote{It is interesting to note that if we use the adiabatic speed of sound $c_0$ instead of $c_s$ then the condition defining the critical temperature becomes $c_0^2=\gamma$; for this parameterization of the equilibrium state, the eigenvalues with zero real part occur only when this relation holds, and $\theta_0$ can take any value.} however, differentiating equation \eqref{eigen_rel_sys} again (twice), this time at $\theta_0 = (\gamma - c_s^2) / \gamma^2$ and with $\lambda = \pm i n \gamma^\frac12$, we find
\bea
\Real \left( \frac{\partial \lambda}{\partial \theta_0} \right) &=& 0; \\
\Real \left( \frac{\partial^2 \lambda}{\partial \theta_0^2} \right) &=& -\frac{n^2 \kappa c_h^4 \gamma^4 \theta_0}{n^2 \kappa^2 (\gamma - c_h^2)^2 + \gamma c_h^4} < 0,
\eea
and so the eigenvalues cannot leave the half-plane $\Real(\lambda)<0$ as $\theta_0$ increases. Consequently, the Fourier coefficients $a_n(t)$, $b_n(t)$ and $d_n(t)$ in \eqref{thetacos}-\eqref{qsin} decay exponentially for $\theta_0 > 0$, except at the critical temperature $\theta_0 = (\gamma - c_s^2) / \gamma^2$ when $\gamma > c_s^2$. 

Given $\delta > 0$, let $N \in \bbN$ be such that
\be
\sum_{n>N} |a_n(0)|, \sum_{n>N} |b_n(0)|, \sum_{n>N} |d_n(0)| < \frac{\delta}2
\ee
(which must exist because the Fourier series of a $C^1$ function is absolutely convergent). Since $|e^{\lambda t}| \leq 1$ for $\Real(\lambda t) \leq 0$, we then have
\be
\sum_{n>N} |a_n(t)|, \sum_{n>N} |b_n(t)|, \sum_{n>N} |d_n(t)| < \frac{\delta}2
\ee
for $t \geq 0$. On the other hand, since each Fourier coefficient decays exponentially, we can choose $T>0$ such that for all $t \geq T$
\be
\sum_{n\leq N} |a_n(t)|, \sum_{n\leq N} |b_n(t)|, \sum_{n\leq N} |d_n(t)| < \frac{\delta}2.
\ee
It is then clear from \eqref{thetacos}-\eqref{qsin} that
\be
|\hat{\theta}(x, t)|, |w(t,x)|, |q(t,x)| < \delta
\ee
for $t \geq T$.
\end{proof}
%
%
\section{Solving the classical system with heat waves}\label{section6}
For comparison purposes, we will now use the Fourier decomposition \eqref{thetacos}-\eqref{qsin} to prove that the solutions of the classical system with heat waves \eqref{heatwaves1}-\eqref{heatwaves3} with boundary conditions \eqref{bd_con} also decay in time. The Fourier modes are easily seen to satisfy the system of ordinary differential equations
\bea
\dot{a}_n(t)&=&-n\gamma\theta_0\dot{b}_n(t) - nd_n(t),\\
\ddot{b}_n(t)&=&-n^2c_s^2b_n(t)+n\gamma a_n(t),\\
\dot{d}_n(t)&=&-\frac{c_h^2}{\kappa}d_n(t)+nc_h^2a_n(t),
\eea
for each $n \in \bbN$. We transform this into a first order system by setting
\bea
\dot{b}_n(t)&=&f_n(t),
\eea
whence
\bea
\lb{abcsyst2}
\dot{a}_n(t)&=&-nd_n(t)-n\gamma\theta_0 f_n(t),\\
\dot{d}_n(t)&=&nc_h^2a_n(t)-\frac{c_h^2}{\kappa}d_n(t), \\
\dot{f}_n(t)&=&n\gamma a_n(t)-n^2 c_s^2b_n(t).
\eea
In matrix form, the system reads
\bea
\begin{bmatrix}
\dot{a}_n  \\ 
\dot{b}_n  \\
\dot{d}_n  \\
\dot{f}_n 
\end{bmatrix}
=
\begin{bmatrix}
0    &0&- n &-n\gamma \theta_0 \\ 
0  & 0&0 &1 \\
n c_h^2 & 0&-\frac{c_h^2}{\kappa}&0 \\
n\gamma &-n^2c_s^2&0&0
\end{bmatrix}
\begin{bmatrix}
{a}_n  \\ 
{b}_n  \\
{d}_n  \\
{f}_n 
\end{bmatrix}.
\eea
The eigenvalues of the system's matrix are obtained by solving the polynomial equation
\bea \label{eigen_heat_wave_sys}
\kappa \lambda^4 + c_h^2 \lambda^3 + n^2\kappa (c_h^2 + c_s^2 + \gamma^2 \theta_0) \lambda^2 + n^2 c_h^2 (c_s^2 + \gamma^2 \theta_0) \lambda + n^4\kappa c_h^2 c_s^2 = 0
\eea
(which can be seen as the limit of \eqref{eigen_rel_sys} when $c_h^2 \theta_0 \ll 1$ and $c_h^2 \ll \gamma$). It is interesting to note that for $\theta_0=0$ the eigenvalues are given by
\bea
\lambda= \frac{-c_h^2\pm c_h \sqrt{c_h^2- 4 n^2\kappa^2}}{2\kappa} \qquad \text{ or } \qquad  \lambda=\pm in c_s.
\eea
These correspond to decaying modes, associated to the diffusion constant $\kappa$ and the speed of heat $c_h$, and non-decaying oscillating modes, associated to the speed of sound $c_s$.

To obtain the first order correction in $\theta_0$ we differentiate equation~\eqref{eigen_heat_wave_sys} at $\theta_0=0$,
\bea
\left[4 \kappa \lambda^3 + 3c_h^2\lambda^2 + 2 n^2\kappa (c_h^2 + c_s^2) \lambda + n^2 c_h^2 c_s^2\right] \frac{\partial \lambda}{\partial \theta_0} + n^2\kappa \gamma^2 \lambda^2 + n^2 c_h^2 \gamma^2 \lambda = 0.
\eea
For $\lambda=\pm in c_s$ we obtain
\bea
\left[- 2 n^2 c_h^2 c_s^2 \pm 2i n^3\kappa c_s (c_h^2 - c_s^2) \right] \frac{\partial \lambda}{\partial \theta_0} - n^4\kappa c_s^2 \gamma^2 \pm i n^3 c_h^2 c_s \gamma^2 = 0,
\eea
and so
\begin{align} \lb{partial_lambda_classical}
\frac{\partial \lambda}{\partial \theta_0} & = \frac{n \gamma^2 (n \kappa c_s \mp i c_h^2)}{- 2 c_h^2 c_s \pm 2i n\kappa (c_h^2 - c_s^2)} \nonumber \\
& = \frac{n \gamma^2 (- n \kappa c_h^4 \pm i c_h^4 c_s \mp i n^2 \kappa^2 c_s (c_h^2 - c_s^2))}{2 [c_h^4 c_s^2 + n^2\kappa^2 (c_h^2 - c_s^2)^2]}.
\end{align}
This result can be obtained from equation \eqref{partial_lambda_rel}, as one might expect, in the limit $c_s^2, c_h^2 \ll \gamma$. Note also that as $n \to \infty$ \eqref{partial_lambda_classical} becomes
\bea
\frac{\partial \lambda}{\partial \theta_0} \sim \mp \frac{i n c_s \gamma^2}{2(c_h^2 - c_s^2)},
\eea
and so
\bea
\lambda \sim \pm i n c_s \left[ 1 - \frac{\gamma^2\theta_0}{2(c_h^2 - c_s^2)} + \ldots \right],
\eea
in agreement with \eqref{c_s_classical}.

\begin{Thm}
All modes in the Fourier expansion \eqref{thetacos}-\eqref{qsin} for the solutions of the classical system with heat waves \eqref{heatwaves1}-\eqref{heatwaves3} with boundary conditions \eqref{bd_con} decay exponentially in time for $\theta_0 > 0$. Consequently, given initial conditions $\hat{\theta}(0,x)$, $w(0,x)$, $\dot{w}(0,x)$ and $q(0,x)$ of class $C^1$, with $\int_0^\pi \hat{\theta}(0,x) dx = 0$, we have
\be
\lim_{t \to +\infty} \hat{\theta}(x, t) = \lim_{t \to +\infty} w(t,x) = \lim_{t \to +\infty} q(t,x) = 0
\ee
for $\theta_0 > 0$.
\end{Thm}

\begin{proof}
Since
\bea
\Real \left( \frac{\partial \lambda}{\partial \theta_0} \right) = - \frac{n^2 \gamma^2 \kappa c_h^4}{2[c_h^4 c_s^2 + n^2\kappa^2 (c_h^2 - c_s^2)^2]} < 0,
\eea
the eigenvalues (which depend continuously on $\theta_0$) enter the half-plane $\Real(\lambda)<0$ as $\theta_0$ increases. If we set $\lambda=iy$ in equation~\eqref{eigen_heat_wave_sys}, with $y \in \bbR$, then we obtain from the imaginary part
\be
y = 0 \qquad \text{ or } \qquad y^2 = n^2(c_s^2+\gamma^2 \theta_0).
\ee
Now equation~\eqref{eigen_heat_wave_sys} is never satisfied in the first case, and is satisfied in the second case if and only if $\theta_0=0$. Therefore, the eigenvalues cannot leave the half-plane $\Real(\lambda)<0$ as $\theta_0$ increases. Consequently, the Fourier coefficients $a_n(t)$, $b_n(t)$ and $d_n(t)$ in \eqref{thetacos}-\eqref{qsin} decay exponentially for $\theta_0 > 0$, and the proof of Theorem~\ref{Thmrel} now applies.
\end{proof}
%
%
\section{Solving the classical system}\label{section7}
For comparison purposes, we will now use the Fourier decomposition \eqref{thetacos}-\eqref{qsin} to prove that the solutions of the classical system \eqref{classical1}-\eqref{classical3} with boundary conditions \eqref{bd_con} also decay in time. This has been done before in \cite{Hansen92}, but expanding about zero coupling parameter $\gamma$, instead of zero equilibrium temperature $\theta_0$. The Fourier modes are easily seen to satisfy the system of ordinary differential equations
\bea
\dot{a}_n(t)&=&-n^2\kappa a_n(t)-n\gamma\theta_0\dot{b}_n(t),\\
\ddot{b}_n(t)&=&-n^2c_s^2b_n(t)+n\gamma a_n(t),
\eea
for each $n \in \bbN$. We transform this into a first order system by setting
\bea
\dot{b}_n(t)&=&f_n(t),
\eea
whence
\bea
\lb{abcsyst_class}
\dot{a}_n(t)&=&-n^2\kappa a_n(t)-n\gamma\theta_0f_n(t),\\
\dot{f}_n(t)&=&n\gamma a_n(t)-n^2c_s^2b_n(t).
\eea
In matrix form, the system reads
\bea
\begin{bmatrix}
\dot{a}_n  \\ 
\dot{b}_n  \\
\dot{f}_n 
\end{bmatrix}
=
\begin{bmatrix}
-n^2\kappa & 0 &-n\gamma \theta_0 \\ 
0  & 0 & 1 \\
n\gamma &-n^2c_s^2&0
\end{bmatrix}
\begin{bmatrix}
{a}_n  \\ 
{b}_n  \\
{f}_n 
\end{bmatrix}.
\eea
The eigenvalues of the system's matrix are obtained by solving the polynomial equation
\bea \label{eigen_classical_sys}
\lambda^3 + n^2\kappa\lambda^2 + n^2 (c_s^2 + \gamma^2 \theta_0) \lambda + n^4\kappa c_s^2 = 0.
\eea
(which can be seen as the limit of \eqref{eigen_heat_wave_sys} as $c_h \to \infty$). It is interesting to note that for $\theta_0=0$ the eigenvalues are given by
\bea
\lambda= - n^2\kappa \qquad \text{ or } \qquad  \lambda=\pm in c_s.
\eea
These correspond to decaying modes, associated to the diffusion constant $\kappa$, and non-decaying oscillating modes, associated to the speed of sound $c_s$.

To obtain the first order correction in $\theta_0$ we differentiate equation~\eqref{eigen_heat_wave_sys} at $\theta_0=0$,
\bea
\left(3\lambda^2 + 2n^2\kappa\lambda + n^2 c_s^2\right) \frac{\partial \lambda}{\partial \theta_0} + n^2 \gamma^2 \lambda = 0.
\eea
For $\lambda=\pm in c_s$ we obtain
\bea
\left(- 2 n^2 c_s^2 \pm 2in^3\kappa c_s \right) \frac{\partial \lambda}{\partial \theta_0} \pm i n^3 c_s \gamma^2 = 0,
\eea
and so
\bea
\frac{\partial \lambda}{\partial \theta_0} = \frac{\mp i n \gamma^2}{- 2 c_s \pm 2in\kappa} = \frac{n \gamma^2 (-n \kappa \pm i c_s)}{2(c_s^2 + n^2\kappa^2)}.
\eea
This result can be obtained from equation \eqref{partial_lambda_classical}, as one might expect, in the limit $c_h \to \infty$. 

\begin{Thm}
All modes in the Fourier expansion \eqref{thetacos}-\eqref{qsin} for the solutions of the classical system \eqref{classical1}-\eqref{classical3} with boundary conditions \eqref{bd_con} decay exponentially in time for $\theta_0 > 0$. Consequently, given initial conditions $\hat{\theta}(0,x)$, $w(0,x)$, $\dot{w}(0,x)$ and $q(0,x)$ of class $C^1$, with $\int_0^\pi \hat{\theta}(0,x) dx = 0$, we have
\be
\lim_{t \to +\infty} \hat{\theta}(x, t) = \lim_{t \to +\infty} w(t,x) = \lim_{t \to +\infty} q(t,x) = 0
\ee
for $\theta_0 > 0$.
\end{Thm}

\begin{proof}
Since
\bea
\Real \left( \frac{\partial \lambda}{\partial \theta_0} \right) = - \frac{n^2 \gamma^2 \kappa}{2(c_s^2 + n^2\kappa^2)} < 0,
\eea
the eigenvalues (which depend continuously on $\theta_0$) enter the half-plane $\Real(\lambda)<0$ as $\theta_0$ increases. If we set $\lambda=iy$ in equation~\eqref{eigen_classical_sys}, with $y \in \bbR$, then we obtain from the imaginary part
\be
y = 0 \qquad \text{ or } \qquad y^2 = n^2(c_s^2+\gamma^2 \theta_0).
\ee
Now equation~\eqref{eigen_classical_sys} is never satisfied in the first case, and is satisfied in the second case if and only if $\theta_0=0$. Therefore, the eigenvalues cannot leave the half-plane $\Real(\lambda)<0$ as $\theta_0$ increases. Consequently, the Fourier coefficients $a_n(t)$, $b_n(t)$ and $d_n(t)$ in \eqref{thetacos}-\eqref{qsin} decay exponentially for $\theta_0 > 0$, and the proof of Theorem~\ref{Thmrel} now applies.
\end{proof}
%
%
\section{Conclusions}\label{section8}
In this paper we have obtained, for the first time, the linear hyperbolic system describing a relativistic heat-conducting rod, accurate to first order in the perturbations around an equilibrium state. We studied some general properties of this system, finding a decreasing energy integral and computing the characteristic propagation speeds. Interestingly, the speed of the mechanical waves is not exactly the nominal speed of sound, which applies to adiabatic processes, since our system features energy dissipation and heat flow. Finally, we proved that the solutions of the system with boundary conditions corresponding to thermal isolation and clamped endpoints decay in time, by means of a Fourier decomposition. Comparing with the classical systems (both with and without heat waves), we found that although the relativistic terms introduce some modifications the qualitative picture remains basically the same. There is an interesting nuance in that the relativistic system may fail to exhibit decay in time at a single critical temperature (for fixed values of the other parameters).
%
%
\section*{Acknowledgments}
This work was partially funded by FCT/Portugal through UID/MAT/04459/2013 and grant (GPSEinstein) PTDC/MAT-ANA/1275/2014.
AF was supported by the FCT scholarship SFRH/BPD/115959/2016.
%
%

\end{document}